%% file: cao_vontobel_itw2017_subm1.tex
\documentclass[conference,10pt]{IEEEtran}
\usepackage[utf8]{inputenc}
\usepackage{amsmath,amssymb,amsthm,amsfonts}
\usepackage{accents}
\usepackage{mathtools,thmtools}
\usepackage{tikz,pgfplots}
\usepackage{bm, bbm}
\usepackage{color,xcolor,colortbl}
\usepackage{watermark}
\usepackage{epsfig}
\usepackage{subfigure}
\usepackage{algorithm,algpseudocode}
\algnewcommand\algorithmicforeach{\textbf{for each}}
\algdef{S}[FOR]{ForEach}[1]{\algorithmicforeach\ #1\ \algorithmicdo}

\pgfplotsset{compat = 1.7}
\usetikzlibrary{positioning,shapes,arrows}
\usetikzlibrary{decorations.markings}

\renewcommand{\mathbf}[1]{{\bm{#1}}}

\newcommand{\eg}{\emph{e.g.}}
\newcommand{\etal}{\emph{et al.}}
\newcommand{\ie}{\emph{i.e.}}

\newcommand{\etc}{\emph{etc.}}

\renewcommand{\geq}{\geqslant}

\newcommand{\set}[1]{\mathcal{#1}}

\newcommand{\defeq}{\triangleq}

\newcommand{\operator}[1]{\tilde{#1}}

\newcommand{\setX}{\set{X}}

\newcommand{\matr}[1]{\mathbf{#1}}
\newcommand{\vect}[1]{\mathbf{#1}}

\newcommand{\vx}{\vect{x}}

\newcommand{\vy}{\vect{y}}

\newcommand{\vxf}{\vx_{\partial{f}}}
\newcommand{\vxpf}{\vx'_{\partial{f}}}
\newcommand{\vyf}{\vy_{\delta{f}}}

\newcommand{\Herm}{\mathsf{H}}

\newcommand{\ZBethe}{Z_{\mathrm{Bethe}}}
\newcommand{\FBethe}{F_{\mathrm{Bethe}}}

\newcommand{\perm}{\operatorname{perm}}

\pgfdeclarelayer{bg}    
\pgfdeclarelayer{fg}    
\pgfsetlayers{bg,main,fg}  

\tikzset{-*/.style={shorten >=#1,decoration={markings,
            mark={at position 1 with {\draw[fill] circle [radius=#1];}}},
          postaction=decorate},-*/.default=1.5pt}
\tikzset{*-/.style={shorten <=#1,decoration={markings,
            mark={at position 0 with {\draw[fill] circle [radius=#1];}}},
          postaction=decorate},*-/.default=1.5pt}
\tikzset{*-*/.style={shorten <=#1,shorten >=#1,decoration={markings,
            mark={at position 0 with {\draw[fill] circle [radius=#1];}},
            mark={at position 1 with {\draw[fill] circle [radius=#1];}},},
          postaction=decorate},*-*/.default=1.5pt}

\theoremstyle{plain}

\theoremstyle{definition}

\theoremstyle{remark}

\newcounter{mytempeqcounter}

\declaretheorem[style=plain]{theorem}

\declaretheorem[style=plain,sibling=theorem]{proposition}
\declaretheorem[style=definition,sibling=theorem,
                                 qed=$\blacksquare$]{definition}
\declaretheorem[style=definition,sibling=theorem,
                                 qed=$\blacktriangle$]{example}
\declaretheorem[style=definition,sibling=theorem, 
                                 qed=$\blacksquare$]{assumption}
\declaretheorem[style=definition,sibling=theorem,
                                 qed=$\blacktriangle$]{remark}

\begin{document}

\title{Double-Edge Factor Graphs: \\
       Definition, Properties, and Examples}

\author{
    \IEEEauthorblockN{Michael~X.~Cao and Pascal~O.~Vontobel}
    \IEEEauthorblockA{Department of Information Engineering \\
                      The Chinese University of Hong Kong \\
                      \{m.x.cao, pascal.vontobel\}@ieee.org\\[-0.30cm]}
}

\maketitle

\begin{abstract}
  Some of the most interesting quantities associated with a factor graph are
  its marginals and its partition sum. For factor graphs \emph{without cycles}
  and moderate message update complexities, the sum-product algorithm
  (SPA) can be used to efficiently compute these quantities exactly. Moreover,
  for various classes of factor graphs \emph{with cycles}, the SPA has been
  successfully applied to efficiently compute good approximations to these
  quantities. Note that in the case of factor graphs with cycles, the local
  functions are usually non-negative real-valued functions.

  In this paper we introduce a class of factor graphs, called double-edge
  factor graphs (DE-FGs), which allow local functions to be complex-valued and
  only require them, in some suitable sense, to be positive semi-definite. We
  discuss various properties of the SPA when running it on DE-FGs and we show
  promising numerical results for various example DE-FGs, some of which have
  connections to quantum information processing.
\end{abstract}

\section{Introduction}
\label{sec:introduction:1}

On the one hand, many classical algorithms like Kalman filtering, the BCJR
algorithm, the forward-backward algorithm, \etc, can be seen as special cases
of the sum-product algorithm (SPA) applied to suitable \emph{cycle-free}
factor graphs~\cite{Kschischang:Frey:Loeliger:01, Loeliger:04:1}. On the other
hand, the SPA has also been successfully applied to various classes of factor
graphs \emph{with cycles}, as is for example witnessed by the SPA-based
decoding techniques of low-density parity-check (LDPC) codes, which appear
nowadays in various telecommunication
standards~\cite{Kschischang:Frey:Loeliger:01, Loeliger:04:1}.

For the case of SPA on factor graphs with cycles, there are a few results that
hold for large classes of factor graphs (like the result by Yedidia
\etal~\cite{Yedidia:Freeman:Weiss:05:1}, which states that fixed points of the
SPA correspond to stationary points of the Bethe free energy function) or the
graph-cover-based interpretation of the Bethe approximation of the partition
sum~\cite{Vontobel:13:2}, but in general the results are for special classes
of factor graphs like Gaussian graphical models (see, \eg,
\cite{Malioutov:Johnson:Willsky:06:1}) or log-supermodular (``attractive'')
graphical models (see, \eg, \cite{Ruozzi:12:1}). In most of these cases, the
focus has been on factor graphs with non-negative real-valued local functions.
However, there are applications, in particular in the area of quantum
information processing, where one would like to have more general factor
graphs. Let us mention some of the approaches that have been pursued:
\begin{itemize}

\item One approach replaces scalar-valued local functions by matrix-valued
  local functions (see, \eg, \cite{Leifer:Poulin:08:1, Cao:Vontobel:16:1}).

\item Another approach keeps scalar-valued local functions, but imposes
  certain symmetry conditions on the factor
  graph~\cite{Loeliger:Vontobel:12:1, Loeliger:Vontobel:15:1:subm}. (See the
  discussion and the references in~\cite{Loeliger:Vontobel:15:1:subm} on how
  the factor graphs therein are related to tensor networks, \etc) The
  framework in~\cite{Loeliger:Vontobel:12:1, Loeliger:Vontobel:15:1:subm} can,
  for example, be conveniently used for estimating information rates of
  channels with a classical input and output and a quantum
  memory~\cite{Cao:Vontobel:17:1}.

\end{itemize}

Note that all marginal calculations were done exactly
in~\cite{Loeliger:Vontobel:12:1, Loeliger:Vontobel:15:1:subm,
  Cao:Vontobel:17:1}. This can be achieved, for example, by first merging
suitable variables so that the resulting factor graph is cycle free and then
to apply the SPA. (Of course, this only gives practical algorithms as long as
the alphabet sizes of the merged variables are not too large.)

However, similar to the above-mentioned classes of factor graphs with cycles,
it is tempting to also apply the SPA to factor graphs as
in~\cite{Loeliger:Vontobel:12:1, Loeliger:Vontobel:15:1:subm} with
cycles. There are different approaches to accomplish this by suitably
reformulating the factor graphs in~\cite{Loeliger:Vontobel:12:1,
  Loeliger:Vontobel:15:1:subm}, some reformulations having better complexity
properties, some reformulations having better analytical properties. An
interesting option in this design space are the double-edge factor graphs
(DE-FGs) that we introduce in this paper.\footnote{As we will see, the name
  ``double-edge'' comes from the fact that pairs of edges (and with that the
  associated variables) are merged. For example, referring to in
  Fig.~\ref{fig:quantum:partial:measurement}~(top), the edge associated with
  variable $x_0$ and the edge associated with variable $x'_0$ are merged to a
  double-edge in Fig.~\ref{fig:quantum:partial:measurement}~(bottom).}

This paper is structured as follows. We define DE-FGs in
Section~\ref{sec:denfg:1} and then formulate the SPA, along with some of its
properties, in Section~\ref{sec:spa}. We discuss a variety of examples in
Section~\ref{sec:examples:1}, we point out connections to a recent paper by
Mori in Section~\ref{sec:connections:mori:1}, and we conclude the paper in
Section~\ref{sec:conclusion:1}. Note that throughout this paper, all alphabets
are assumed to be finite.

\section{Double-edge Factor Graphs}
\label{sec:denfg:1}

In this section we define double-edge factor graphs (DE-FGs), more precisely,
double-edge normal factor graphs (DE-NFGs). The word ``normal'' refers to the
fact that variables appear as arguments of only one or two local
functions.\footnote{In the same way that any factor graph can be suitably
  reformulated as a normal factor graph~\cite{Forney:01:1}, any DE-FG can be
  suitably reformulated as a DE-NFG. With this, there is no loss in generality
  in considering only DE-NFGs.}

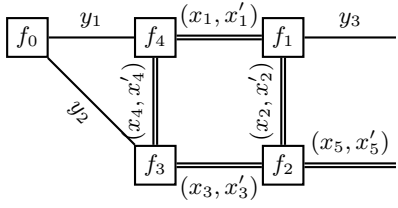
\begin{figure}
  \centering
  \input{figs2/demo.tikz}
  \caption{A double-edge normal factor graph (DE-NFG).}
  \label{fig:demo}
\end{figure}

\begin{example}
  Consider the DE-NFG in Fig.~\ref{fig:demo}, which is a pictorial
  representation of the factorization
  \begin{align*}
    g(\vx,\vx';\vy)
    &= f_0(y_1,y_2)
         \nonumber \\
    &\hspace{-0.5cm}
       \cdot \,
       f_1(x_1,x_2,x_1',x_2';y_3) 
       \cdot 
       f_2(x_2,x_3,x_5,x_2',x_3',x_5') 
         \nonumber \\
    &\hspace{-0.5cm}
       \cdot \,
       f_3(x_3,x_4,x_3',x_4';y_2)
       \cdot 
       f_4(x_4,x_1,x_4',x_1';y_1).
  \end{align*}
  It is called a DE-NFG because some of the edges are double lines that
  correspond to variables that are paired. (For example, $x_1$ and $x'_1$ are
  paired in Fig.~\ref{fig:demo}.) Such paired variables are assumed to have
  both the same alphabet. Moreover, as detailed below, the local functions
  have to satisfy some constraints.
\end{example}

\begin{definition}
  Consider the factorization
  \begin{align*}
    g(\vx,\vx';\vy) 
      &= \prod_{f\in\set{F}}
           f(\vxf,
             \vxpf;
             \vyf)
  \end{align*}
  represented by some DE-NFG. We will use the following conventions:
  \begin{itemize}
  
  \item We call $g$ the global function.
  
  \item We call $f \in \set{F}$ the local functions. With some abuse of
    notation, we will also use $f$ to refer to the corresponding function node
    in the DE-NFG.
  
  \item For every function node $f \in \set{F}$, the variables associated with
    the incident double-edges are collected in $\vxf, \vxpf$.
  
  \item For every function node $f \in \set{F}$, the variables associated with
    the incident single-edges are collected in $\vyf$.
  
  \end{itemize}
  Most importantly, we require every local function $f \in \set{F}$ to have
  the following property:
  \begin{center}
    the local function $f(\vxf, \vxpf; \vyf)$ is complex-valued \\
    and is positive semi-definite (PSD).
  \end{center}
  The latter property is to be understood as follows: for every $\vyf$ and
  every complex-valued function $h$ over the alphabet of $\vxf$ (and with that
  also over the alphabet of $\vxpf$), it holds that
  \begin{align}
    \sum_{\vxf, \, 
          \vxpf}
      \overline{h(\vxf)}
      \cdot
      f(\vxf, 
        \vxpf;
        \vyf)
      \cdot
      h(\vxpf)
      &\geq 
         0.
           \label{eq:denfg:psd:condition:1}
  \end{align}
  (Here and in the following, over-bar denotes complex conjugation.) Clearly,
  if a function node $f$ has no incident double edges, then the condition
  in~\eqref{eq:denfg:psd:condition:1} reduces to the condition that the local
  function $f$ takes on only non-negative real values.
\end{definition}

For proving various properties of DE-NFG, the following observation is very
beneficial.

\begin{remark}
  \label{remark:psd:decomposition:1}

  For every local function $f \in \set{F}$ and every $\vyf$, there are a
  finite set $\set{K}_{f,\vyf}$ and some complex-valued functions
  $b_{f,\vyf,k}$ \!, $k \in \set{K}_{f,\vyf}$ \!, over the alphabet of
  $\vxf$ such that
  \begin{align*}
    f(\vxf, \vxpf; \vyf) 
      &= \sum_{k \in \set{K}_{f,\vyf}}
           \overline{b_{f,\vyf,k}(\vxf)}
           \cdot
           b_{f,\vyf,k}(\vxpf) \ .
  \end{align*}
  This follows easily from the eigenvalue decomposition of PSD matrices.
\end{remark}

\begin{proposition}
  The partition sum of a DE-NFG, \ie,
  \begin{align*}
    Z 
      &\defeq 
         \sum_{\vx, \, \vx', \, \vy}
           g(\vx,\vx';\vy) \ , 
  \end{align*}
  is a non-negative real number.
\end{proposition}

\begin{proof}
  This can be proven with the help of
  Remark~\ref{remark:psd:decomposition:1}. We omit the details because of
  space limitations.
\end{proof}

\begin{figure}[t]
  \centering
  \input{figs2/partial_measurement_newnew.tikz}
  \vskip0.5cm
  \input{figs2/partial_measurement_defg_newnew.tikz}
  \caption{Top: NFG for Example~\ref{ex:quantum:measurement}. Bottom: DE-NFG
    for Example~\ref{ex:quantum:measurement}.}
  \label{fig:quantum:partial:measurement}
\end{figure}
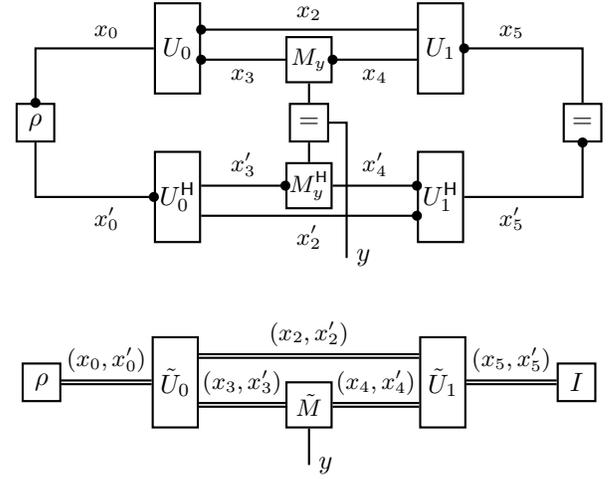

As already mentioned, one of the main motivations of the present paper are the
NFGs in~\cite{Loeliger:Vontobel:12:1, Loeliger:Vontobel:15:1:subm}. So let us
show how a ``typical'' NFG in~\cite{Loeliger:Vontobel:12:1,
  Loeliger:Vontobel:15:1:subm} can be formulated as a DE-NFG.

\begin{example}
  \label{ex:quantum:measurement}

  Consider the NFG in Fig.~\ref{fig:quantum:partial:measurement}~(top), which
  can be used to do probability computations for the following quantum
  mechanical setup:
  \begin{itemize}

  \item At the beginning, some quantum mechanical system is in some mixed
    state (represented by the density matrix $\rho$, which is a PSD matrix).

  \item The system then evolves unitarily (represented by $U_0$).

  \item Afterwards, a sub-system is measured (represented by measurement
    operators $\{ M_{y} \}_{y}$).

  \item Finally the system evolves unitarily (represented by $U_1$).

  \end{itemize}
  (For further details, we refer to~\cite{Loeliger:Vontobel:12:1,
    Loeliger:Vontobel:15:1:subm}.) This NFG can be turned into the DE-NFG
  shown in Fig.~\ref{fig:quantum:partial:measurement}~(bottom) by suitably
  merging edges (and with that the associated variables) and by suitably
  defining the DE-NFG's function nodes. For example, the function node
  $\operator{M}$ is defined to be
  \begin{align}
    \operator{M}(x_3,x_4,x'_3,x'_4; y)
      &\defeq
         M_{y}(x_3,x_4)
         \cdot 
         M_{y}(x'_3,x'_4) \ .
           \label{eq:measurment:function:node:merging:1}
  \end{align}
  Clearly, the function $\operator{M}(x_3,x_4,x'_3,x'_4; y)$ satisfies the
  required PSD constraint. In fact, the expression
  in~\eqref{eq:measurment:function:node:merging:1} is in the form of the
  decomposition in Remark~\ref{remark:psd:decomposition:1}.
\end{example}

One can check that the redrawing procedure in
Example~\ref{ex:quantum:measurement} can be applied to all relevant NFGs
in~\cite{Loeliger:Vontobel:12:1, Loeliger:Vontobel:15:1:subm}.

\clearpage

\section{Sum-Product Algorithm on DE-NFGs \\
               and the Bethe Approximation}
\label{sec:spa}

In this section we define the SPA for DE-NFGs and discuss some of its
properties. In particular, we connect it to generalized versions of the Bethe
free energy function.

Once a DE-NFG as in Fig.~\ref{fig:demo} or in
Fig.~\ref{fig:quantum:partial:measurement}~(bottom) has been defined, we
simply consider it as a particular type of NFG and apply the SPA in the
standard way~\cite{Kschischang:Frey:Loeliger:01, Loeliger:04:1}.  Some
comments:
\begin{itemize}

\item In this paper we only discuss the flooding
  schedule~\cite{Kschischang:Frey:Loeliger:01}, where all messages are updated
  at every iteration. Clearly, other update schedules are possible and might
  be preferable in some cases.

\item If desired, message can be rescaled by a positive scalar at every
  iteration.

\item For reasons of simplicity, we discuss only the case where all edges are
  full edges, \ie, connect two function nodes. (Note that any DE-NFG can be
  turned into such a DE-NFG by attaching suitable dummy function nodes to
  half-edges, thereby turning half-edges into full-edges without changing
  marginals or the partition sum.)

\end{itemize}

Recall that in the case of NFGs, messages are functions over the alphabet of
the variable associated with an edge. Therefore, along a single-edge $e$
between some function nodes $f$ and $h$, we will have messages $\mu^{(t)}_{e
  \to f}(y_e)$ and $\mu^{(t)}_{e \to h}(y_e)$ at iteration $t$. Similarly,
along a double-edge between some function nodes $f$ and $h$, we will have
messages $\mu^{(t)}_{e \to f}(x_e,x'_e)$ and $\mu^{(t)}_{e \to h}(x_e,x'_e)$
at iteration $t$.

\begin{assumption}
  \label{assumption:spa:initialization:1}

  We make the following assumptions about the initial messages, \ie, about the
  messages at time $t = 0$:
  \begin{itemize}

  \item Messages along single-edges are positive real-valued functions.

  \item Messages along double-edges are complex-valued positive definite (PD)
    functions.

  \end{itemize}
  \vskip-0.5cm
\end{assumption}

\begin{proposition}
  \label{prop:spa:behavior:1}

  Let the messages be initialized as in
  Assumption~\ref{assumption:spa:initialization:1}. Then for every iteration
  $t \geq 1$ it holds that:
  \begin{itemize}

  \item Messages along single-edges are non-negative real-valued functions.

  \item Messages along double-edges are complex-valued PSD functions.

  \end{itemize}
\end{proposition}

\begin{proof}
  One approach to prove these statements is based on
  Remark~\ref{remark:psd:decomposition:1}. Another approach is based on
  Schur's product theorem, which states that the component-wise product of two
  PSD matrices is a PSD matrix.\footnote{Actually, Schur's product theorem
    makes the stronger statement that the component-wise product of two PD
    matrices is a PD matrix.}
\end{proof}

\begin{definition}
  \label{def:pseudo:dual:Bethe:parition:sum:1}

  Consider a collection of SPA messages, one for every edge in both
  directions. Let 
  \begin{align*}
    \ZBethe 
      &\defeq 
         \prod_{f \in \set{F}} Z_f 
         \ \Big/ \ 
         \prod_{e \in \set{E}} Z_e 
         \Big. \ ,
  \end{align*}
  where $\set{E}$ is the set of all edges, where for every $f \in \set{F}$ we
  define $Z_f \defeq \sum_{\vxf, \vxpf, \vyf} f(\vxf, \vxpf; \vyf) \cdot
  \bigl( \prod_{e \in \partial{f}} \mu_{e \to f}(x_e,x'_e) \bigr) \cdot \bigl(
  \prod_{e \in \delta{f}} \mu_{e \to f}(y_e) \bigr)$, where for every
  single-edge $e \in \set{E}$ between function nodes $f$ and $h$ we define
  $Z_e \defeq \sum_{y_e} \mu_{e \to f}(y_e) \cdot \mu_{e \to h}(y_e)$, and
  where for every double-edge $e \in \set{E}$ between function nodes $f$ and
  $h$ we define $Z_e \defeq \sum_{x_e, \, x'_e} \mu_{e \to f}(x_e,x'_e) \cdot
  \mu_{e \to h}(x_e,x'_e)$.
\end{definition}

\begin{proposition}
  The function $\ZBethe$ in
  Definition~\ref{def:pseudo:dual:Bethe:parition:sum:1} has the following
  properties:
  \begin{itemize}
    
  \item Assume that the messages have the properties in
    Proposition~\ref{prop:spa:behavior:1} and assume that $\ZBethe$ is
    well-defined, \ie, $Z_e \neq 0$ for all $e \in \set{E}$. Then $\ZBethe$ is
    a non-negative real number.

  \item Fixed points of the SPA correspond to stationary points of the
    function $\ZBethe$. (This generalizes a theorem by Yedidia
    \etal~\cite{Yedidia:Freeman:Weiss:05:1}.)

  \end{itemize}
\end{proposition}

\begin{proof}
  Omitted due to space limitations. 
\end{proof}

Evaluating $\ZBethe$ in Definition~\ref{def:pseudo:dual:Bethe:parition:sum:1}
at a fixed point of the SPA results in the Bethe approximation of the 
partition sum of the DE-NFG.

One can also generalize the Bethe free energy function $\FBethe$
from~\cite{Yedidia:Freeman:Weiss:05:1}, where $\FBethe$ is a function over a
suitable generalization of the local marginal polytope. While a statement
(analogous to a statement in~\cite{Yedidia:Freeman:Weiss:05:1}) that fixed
points of the SPA correspond to stationary points of $\FBethe$ can be made,
evaluating $\ZBethe$ based on $\FBethe$ is trickier because of the
multi-valuedness of the complex logarithm.

\section{Examples}
\label{sec:examples:1}

In this section we discuss various examples of DE-NFGs. In particular, we
compare the obtained Bethe approximation of the partition sum with the true
partition sum. (The NFGs in this section have modest sizes so that the true
partition function can be computed efficiently.) Moreover, for the first
example, we can also make some analytical statements.

\begin{figure}
  \begin{center}
    \subfigure[]{
      \epsfig{file=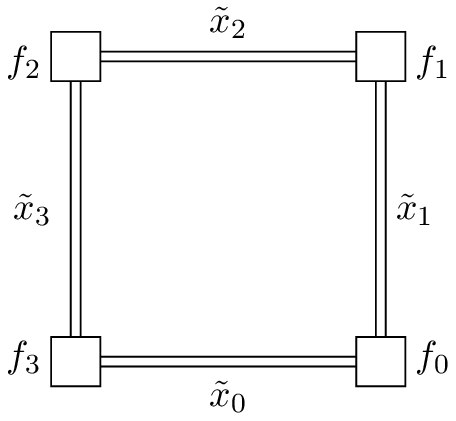,
              width=0.35\linewidth}
    }
    \hspace{0.10cm}
    \subfigure[]{
      \epsfig{file=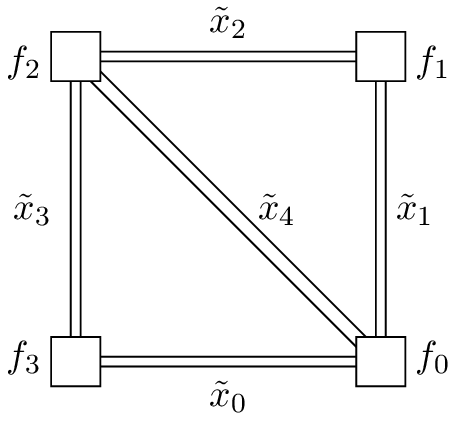,
              width=0.35\linewidth}
    }

    \subfigure[]{
      \parbox{0.45\linewidth}{%
      \includegraphics[width=\linewidth]
                      {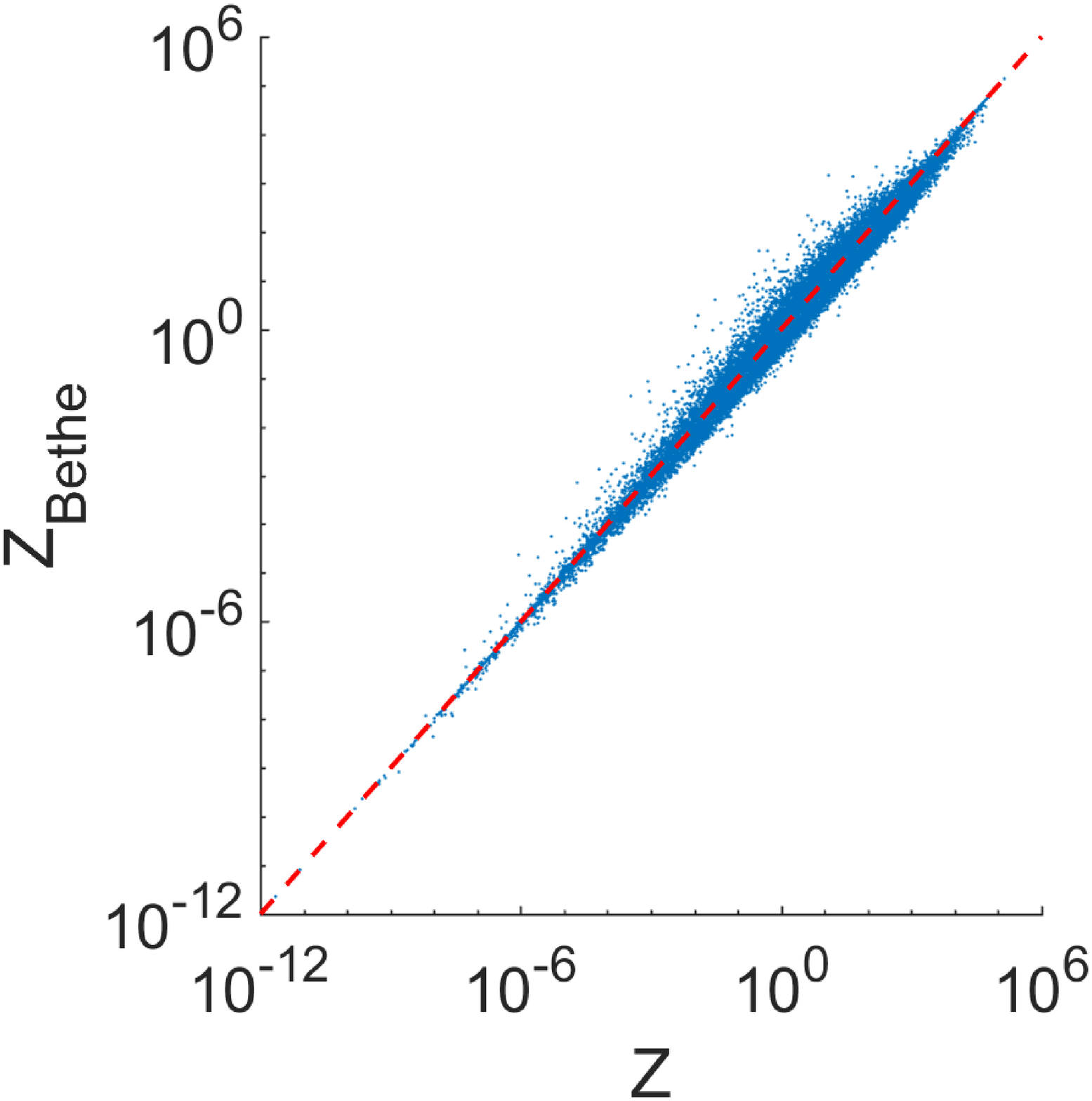}
      \vskip0.25cm
      }
    }
    \hspace{0.10cm}
    \subfigure[]{
      \parbox{0.45\linewidth}{%
      \includegraphics[width=\linewidth]
                      {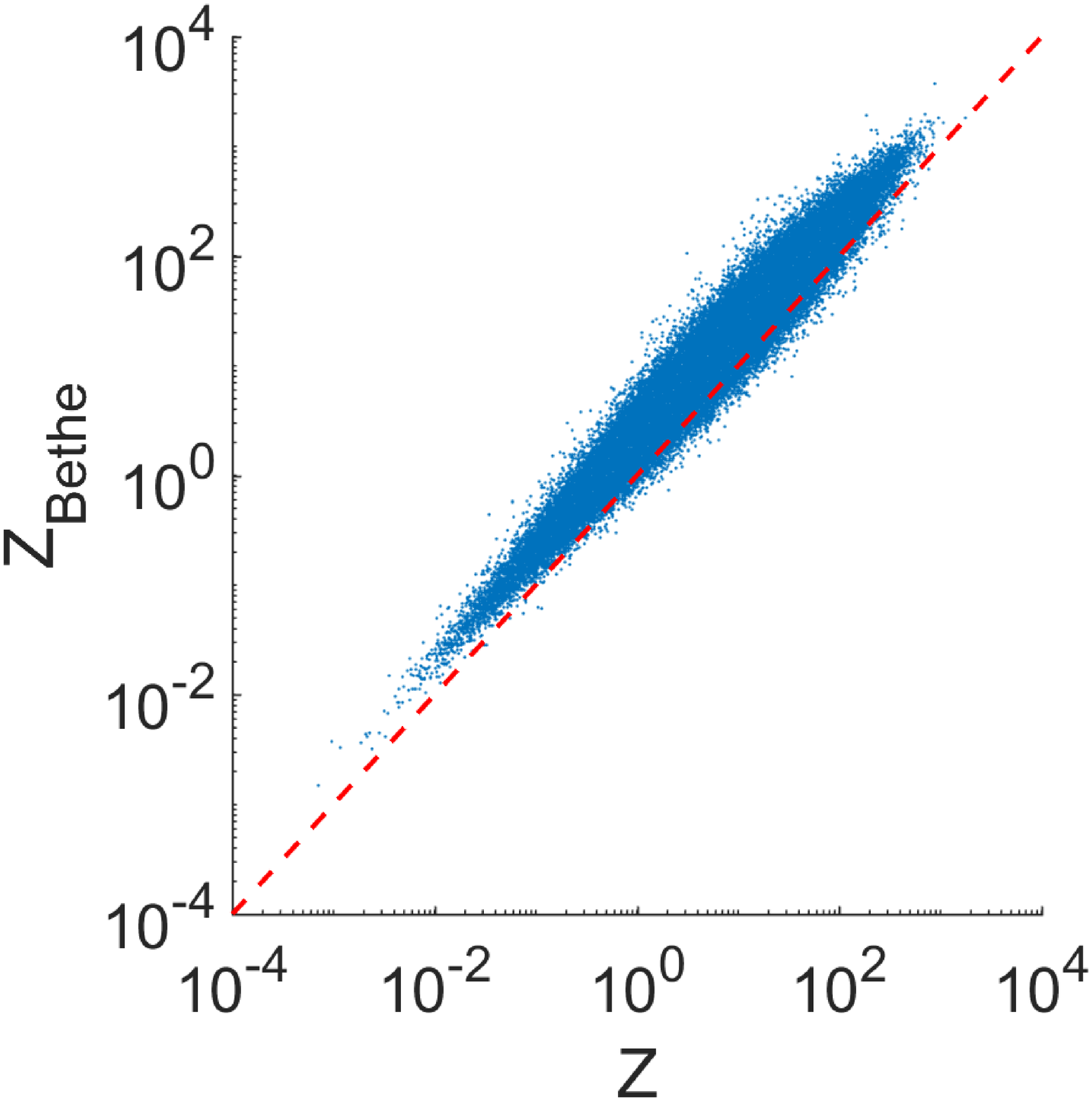}
      \vskip0.25cm
      }
    }
  \end{center}
  \caption{DE-NFGs used in Examples~\ref{example:denfg:four:cycle:1}
    and~\ref{example:denfg:four:cycle:with:diagonal:1}.}
  \label{fig:denfg:length:4:cycle:1}
\end{figure}

\begin{example}
  \label{example:denfg:four:cycle:1}

  Let $n$ be some integer larger than one. Consider a DE-NFG whose topology is
  an $n$-cycle and where all variables take on values in the same finite
  alphabet $\set{X}$. (Fig.~\ref{fig:denfg:length:4:cycle:1} shows such a
  DE-NFG for $n = 4$.) Let $\matr{F}$ be a complex-valued PD matrix of size
  $|\set{X}|^2 \times |\set{X}|^2$ with entries $F(x_0,x_1;x'_0,x'_1)$. For $i
  \in [n] \defeq \{ 0, 1, \ldots, n \! - \! 1 \}$, we define the local
  function $f_i$ to be $f_i(x_i,x_{i+1}; x'_i,x'_{i+1}) \defeq
  F(x_i,x_{i+1};x'_i,x'_{i+1})$. (All indices are modulo $n$.)

  In order to proceed, it is convenient to define the complex-value matrix
  $\matr{B}$ of size $|\set{X}|^2 \times |\set{X}|^2$ with entries
  $B(x_0,x'_0;x_1,x'_1) \defeq F(x_0,x_1;x'_0;x'_1)$ and to define $\tilde x_i
  \defeq (x_i, x'_i)$, $i \in [n]$.

  Let $\mu_{f_i \to f_{i+1}}^{(t)}(\tilde x_{i+1})$ be the SPA message along
  the double edge from $f_i$ to $f_{i+1}$ at time index $t$. Similarly, let
  $\mu_{f_{i+1} \to f_i}^{(t)}(\tilde x_{i+1})$ be the SPA message along the
  double edge from $f_{i+1}$ to $f_i$ at time index $t$. Clearly,
  \begin{align}
    \mu_{f_i \to f_{i+1}}^{(t)}(\tilde x_{i+1})
      &\propto
         \sum_{\tilde x_{i-1}}
           \mu_{f_{i-1} \to f_i}^{(t-1)}(\tilde x_{i-1})
           \cdot
           B(\tilde x_{i-1}, \tilde x_i) \ ,
             \label{eq:denfg:length:4:cycle:spa:1} \\
    \mu_{f_i \to f_{i-1}}^{(t)}(\tilde x_i)
      &\propto
         \sum_{\tilde x_{i+1}}
           B(\tilde x_i, \tilde x_{i+1})
           \cdot
           \mu_{f_{i+1} \to f_i}^{(t-1)}(\tilde x_{i+1}) \ .
             \label{eq:denfg:length:4:cycle:spa:2}
  \end{align}
  For $i \in [n]$, we assume the following initializations $\mu_{f_i \to
    f_{i+1}}^{(0)}(\tilde x_{i+1}) \defeq \delta(x_{i+1},x'_{i+1})$ and
  $\mu_{f_{i+1} \to f_i}^{(0)}(\tilde x_{i+1}) \defeq
  \delta(x_{i+1},x'_{i+1})$, where $\delta$ is the Kronecker-delta function.

  Because of the properties of the matrix $\matr{B}$ that are induced by the
  properties of the matrix $\matr{F}$, the SPA message update rules
  in~\eqref{eq:denfg:length:4:cycle:spa:1}--\eqref{eq:denfg:length:4:cycle:spa:2}
  represent so-called completely positive maps (see, \eg,
  \cite{Nielsen:Chuang:00:1}). (For this statement we ignore the rescaling
  factors.) Using generalizations of Perron--Frobenius theory
  (see~\cite{Evans:HoeghKrohn:78:1, Schrader:01:1}), one can make the
  following statements:
  \begin{itemize}
      
  \item For every $i \in [n]$, the message $\mu_{f_i \to f_{i+1}}^{(t)}(\tilde
    x_{i+1})$ converges to a PD matrix as $t \to \infty$.

  \item For every $i \in [n]$, the message $\mu_{f_{i+1} \to f_i}^{(t)}(\tilde
    x_{i+1})$ converges to a PD matrix as $t \to \infty$.

  \item The eigenvalue of the matrix $\matr{B}$ with maximum absolute value is
    a real number and is unique. Let us call it $\lambda_0$.

  \item The Bethe approximation of the partition sum is
    \begin{align}
      \ZBethe 
        &= \lambda_0^n \ .
    \end{align}

  \end{itemize}
  Compare this result with the partition sum, which is
  \begin{align}
    Z 
      &= \sum_{j=0}^{|\setX|^2-1}
           \lambda^n_j
       = \lambda_0^n
         \cdot
         \left(
           1
           +
           \sum_{j=1}^{|\setX|^2-1}
             \left(
               \frac{\lambda_j}{\lambda_0}
             \right)^n
         \right) \ ,
  \end{align}
  where $\lambda_0, \ldots, \lambda_{|\setX|^2-1}$ are the eigenvalues of
  $\matr{B}$. We see that the smaller the ratios $\bigl(
  \frac{\lambda_j}{\lambda_0} \bigr)^n$, $j = 1, \ldots, |\setX|^2 \! - \! 1$,
  are, the better the Bethe approximation is.

  For $n = 4$ and $|\set{X}| = 2$, Fig.~\ref{fig:denfg:length:4:cycle:1}(c)
  shows the obtained $Z$ and $\ZBethe$ values for $10^6$ experiments based on
  randomly generating matrices $\matr{F} \defeq \matr{U} \cdot \matr{D} \cdot
  \matr{U}^\Herm$, which are based on randomly generating unitary matrices
  $\matr{U}$ and diagonal matrices $\matr{D}$, where the diagonal entries of
  $\matr{D}$ are sampled i.i.d.\ from a standard $\chi^2$ distribution with
  one degree of freedom. We see that very often the ratio $\ZBethe / Z$ is
  rather close to $1$.
\end{example}

\begin{example}
  \label{example:denfg:four:cycle:with:diagonal:1}

  Consider now the DE-NFG in Fig.~\ref{fig:denfg:length:4:cycle:1}(b).  For
  $|\set{X}| = 2$, Fig.~\ref{fig:denfg:length:4:cycle:1}(d) shows the obtained
  $Z$ and $\ZBethe$ values for $10^6$ experiments based on randomly generating
  local functions. In contrast to Example~\ref{example:denfg:four:cycle:1},
  where for every instantiation all local function were the same, here for
  every instantiation all local function are generated independently. We
  observe the ratio $\ZBethe / Z$ is reasonably close to $1$, but typically
  larger than $1$.
\end{example}

\begin{figure}
  \begin{center}
    \subfigure[]{
      \epsfig{file=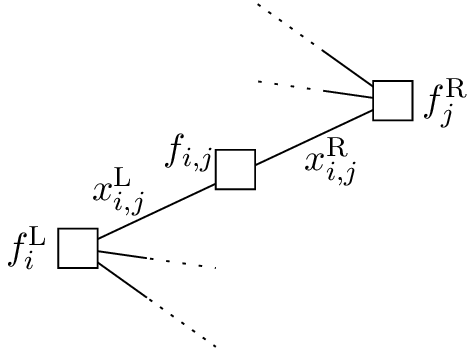,
              width=0.40\linewidth}
    }
    \quad
    \subfigure[]{
      \epsfig{file=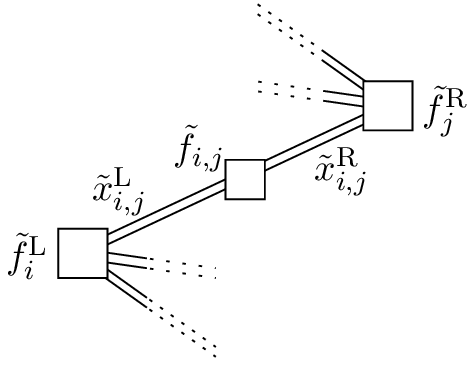,
              width=0.40\linewidth}
    }

    \subfigure[]{
      \parbox{0.38\linewidth}{%
      \includegraphics[width=\linewidth]
        {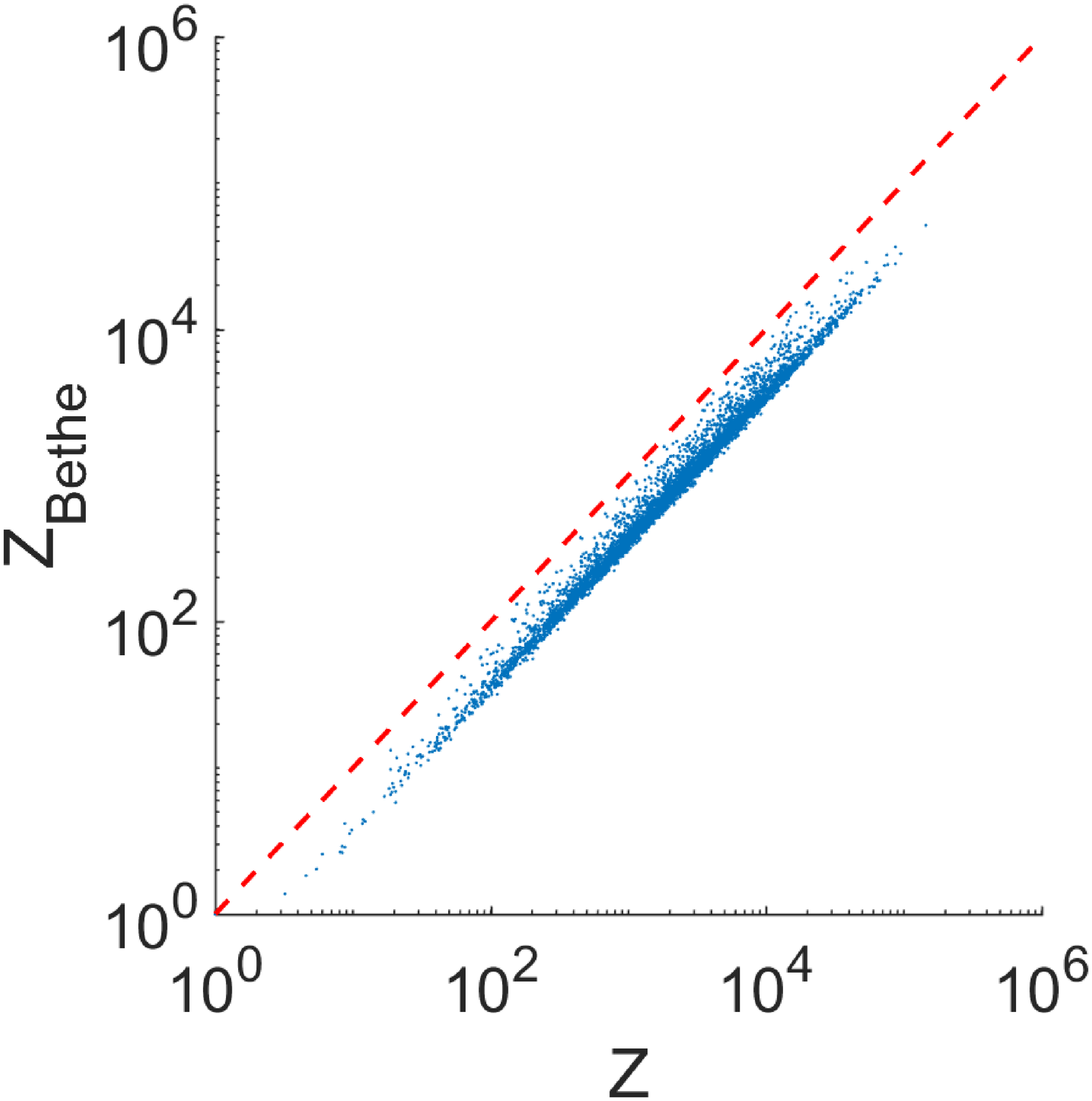} \\
      }
    }
  \end{center}
  \caption{DE-NFGs used in Example~\ref{example:denfg:permanent:1}.}
  \label{fig:denfg:permanent:1}
\end{figure}

\begin{example}
  \label{example:denfg:permanent:1}

  Let $\matr{\theta}$ be a complex-valued matrix of size $n \times n$ with
  entries $\theta_{i,j}$. The permanent~\cite{Minc:78} of $\matr{\theta}$ is
  defined to be $\perm(\matr{\theta}) \defeq \sum_{\sigma} \prod_{i=1}^{n}
  \theta_{i,\sigma(i)}$ \!,
  where the summation is over all $n!$ permutations of the set $[n] \defeq \{
  1, \ldots, n \}$. Ryser's algorithm, one of the most efficient algorithms
  for exactly computing $\perm(\matr{\theta})$ for general matrices $\theta$,
  requires $\Theta(n \cdot 2^n)$ arithmetic operations~\cite{Ryser:63:1}, and
  so the exact computation of permanent is intractable, even for moderate
  values of $n$. Note that even the computation of the permanent of matrices
  that contain only zeros and ones is \#P-complete~\cite{Valiant:79:1}.

  One can formulate an NFG whose partition sum equals $\perm(\matr{\theta})$,
  see, \eg, Fig.~1 in~\cite{Vontobel:13:1}. That NFG is a complete bipartite
  graph with $n$ function nodes on the left and $n$ function nodes on the
  right. Here, Fig.~\ref{fig:denfg:permanent:1}(a), shows a slightly modified
  version of that NFG. All variables take values in the set $\setX = \{ 0, 1
  \}$. Moreover, for all $i \in [n]$, the function $f^{\mathrm{L}}_{i}$ is
  defined to be
  \begin{align*}
    f^{\mathrm{L}}_{i}
      \bigl( 
        \{ x^{\mathrm{L}}_{i,j} \}_{j \in [n]}
      \bigr)
      &\defeq
         \begin{cases}
           1 & \text{exactly one of $\{ x^{\mathrm{L}}_{i,j} \}_{j \in [n]}$ 
                     equals $1$} \\
           0 & \text{otherwise}
         \end{cases} ;
  \end{align*}
  for all $j \in [n]$, the function $f^{\mathrm{R}}_{j}$ is defined
  analogously;
  and for all $(i,j) \in [n]^2$, the function $f_{i,j}$ is defined to be
  \begin{align*}
    f_{i,j}
      \bigl(
        x^{\mathrm{L}}_{i,j}, 
        x^{\mathrm{R}}_{i,j}
      \bigr)
      &\defeq
         \delta
           \bigl(
             x^{\mathrm{L}}_{i,j}, 
             x^{\mathrm{R}}_{i,j}
           \bigr)
         \cdot
         \begin{cases}
           \theta_{i,j} & \text{if $x^{\mathrm{L}}_{i,j} = 1$} \\
           1           & \text{if $x^{\mathrm{L}}_{i,j} = 0$}
         \end{cases} \ .
  \end{align*}

  In this example, we consider the following, rather natural generalization to
  the DE-NFG in Fig.~\ref{fig:denfg:permanent:1}(b), where we will use the
  short-hand $\tilde x^{\mathrm{L}}_{i,j}$ for $\bigl( x^{\mathrm{L}}_{i,j},
  \tilde x^{\mathrm{L'}}_{i,j} \bigr)$, \etc\ Assume that for $(i,j) \in
  [n]^2$, \ $\tilde \theta_{i,j}$ is a complex-valued PSD matrix of size $2
  \times 2$. With this, for $i \in [n]$, the function $\tilde
  f^{\mathrm{L}}_{i}$ is defined to be
  \begin{align*}
    \tilde f^{\mathrm{L}}_{i}
      \bigl( 
        \{ \tilde x^{\mathrm{L}}_{i,j} \}_{j \in [n]}
      \bigr)
      &\defeq
         f^{\mathrm{L}}_{i}
           \bigl( 
             \{ x^{\mathrm{L}}_{i,j} \}_{j \in [n]}
           \bigr)
         \cdot
         f^{\mathrm{L}}_{i}
           \bigl( 
             \{ x^{\mathrm{L}'}_{i,j} \}_{j \in [n]}
           \bigr) \ ;
  \end{align*}
  for all $j \in [n]$, the function $\tilde f^{\mathrm{R}}_{j}$ is defined
  analogously;
  and for all $(i,j) \in [n]^2$, the function $\tilde f_{i,j}$ is defined to be
  \begin{align*}
    \tilde f_{i,j}
      \bigl(
        \tilde x^{\mathrm{L}}_{i,j}, 
        \tilde x^{\mathrm{R}}_{i,j}
      \bigr)
      &\defeq
         \delta
           \bigl(
             x^{\mathrm{L}}_{i,j}, 
             x^{\mathrm{R}}_{i,j}
           \bigr)
         \cdot
         \delta
           \bigl(
             x^{\mathrm{L}'}_{i,j}, 
             x^{\mathrm{R}'}_{i,j}
           \bigr)
         \cdot
         \tilde \theta_{i,j}
           \bigl(
             x^{\mathrm{L}}_{i,j},
             x^{\mathrm{L}'}_{i,j}
           \bigr) \ .
  \end{align*}
  (One can easily verify that these local function define indeed a DE-NFG.)
  Finally, let $\tilde Z$ be the partition sum of this DE-NFG.

  This DE-NFG definition has the following two important special cases:
  \begin{itemize}

  \item If $\tilde \theta_{i,j} = \bigl( \begin{smallmatrix} 1 & 0 \\
      0 & \theta_{i,j} \end{smallmatrix} \bigr)$ for all $(i,j) \in [n]^2$,
    then $\tilde Z = \perm(\matr{\theta})$.

  \item If
    $\tilde \theta_{i,j} 
        = \bigl(
            \begin{smallmatrix} 
              1 \\
              \theta_{i,j} 
            \end{smallmatrix}
          \bigr)
          \cdot 
          \bigl(
            \begin{smallmatrix} 
              1 & \overline{\theta_{i,j}}
            \end{smallmatrix}
          \bigr)
        = \Bigl(
            \begin{smallmatrix} 
              1 & \overline{\theta_{i,j}} \\ 
              \theta_{i,j} & |\theta_{i,j}|^2 
            \end{smallmatrix}
            \Bigr)$ for all $(i,j) \in [n]^2$, then $\tilde Z =
            \perm(\matr{\theta}) \cdot \perm(\overline{\matr{\theta}}) =
            \bigl| \perm(\matr{\theta}) \bigr|^2$, where
            $\overline{\matr{\theta}}$ denotes the matrix whose entries are
            the complex-conjugate values of the entries of
            $\matr{\theta}$. (Note that such partition sums are of interest in
            quantum information processing~\cite{Aaronson:Arkhipov:13:1},
            where $\matr{\theta}$ are certain types of square matrices over
            the complex numbers. We refer to~\cite{Aaronson:Arkhipov:13:1} for
            details.)

  \end{itemize}

  In our experiments, we considered the following setup. Namely, for every
  $(i,j) \in [n]^2$, we independently generate $\tilde \theta_{i,j}$ as
  follows: $\tilde \theta_{i,j}(0,0) \defeq 1$; \ $\tilde \theta_{i,j}(1,0)$
  is picked uniformly from the unit circle in the complex plane; \ $\tilde
  \theta_{i,j}(0,1) \defeq \overline{\tilde \theta_{i,j}(1,0)}$; \ $\tilde
  \theta_{i,j}(1,1)$ is picked uniformly (and independently of the other
  entries) from the real line interval $[1.10, \
  11.10]$. Fig.~\ref{fig:denfg:permanent:1}(c) shows the obtained $Z$ and
  $\ZBethe$ values for $5000$ experiments for the case $n = 5$. We observe
  that the ratio $\ZBethe / Z$ is concentrated around a value smaller than
  $1$.
\end{example}

\section{Connections to a Paper by Mori}
\label{sec:connections:mori:1}

Finally, let us point out that there are strong connections of DE-NFGs to the
setup in Section~V of a recent paper by Mori~\cite{Mori:15:2}. Assume to have
a bipartite DE-NFG. (Such a DE-NFG can always be obtained by suitably
inserting dummy function nodes.) Then the partition sum can be written as some
inner product between, on the one hand, the tensor product of the local
functions corresponding to first class of function nodes of the bipartite
DE-NFG, and, on the other hand, the tensor product of the local functions
corresponding to the second class of function nodes of the bipartite
DE-NFG. Once this connection is observed, one can translate Mori's results
(like loop calculus expansions) to DE-NFGs.

\section{Conclusion}
\label{sec:conclusion:1}

In this paper we have defined DE-NFGs and studied some of their properties. In
particular, we have shown some promising numerical studies of the Bethe
approximation to the partition sum. Many open questions remain. For example,
can some of the results in~\cite{Vontobel:13:1} be generalized to the setup in
Example~\ref{example:denfg:permanent:1}? Or, as in the context of computing
the pattern maximum likelihood estimate, which can be formulated as optimizing
the parameters of some graphical model toward maximizing the partition
function, and where the Bethe partition sum was beneficially used as a
surrogate function~\cite{Vontobel:12:1}, can the Bethe partition sum of a
DE-NFG serve as a suitable surrogate function in some partition function
optimization problem?

\bibliographystyle{IEEEtran}
\bibliography{/home/vontobel/references/references,references_michael1}

\end{document}

%% file: figs2/demo.tikz
\begin{tikzpicture}[nodes = {draw= none, minimum size = 0pt},
	factor/.style={rectangle, minimum size=.5cm, draw, thick},
	node distance=1.7cm,
	every path/.style={every node/.style={font=\small, above=0pt,sloped}}]
	\node[factor] (A) {$f_0$};
	\node[factor, right of = A] (B) {$f_4$};
	\node[factor, right of = B] (C) {$f_1$};
	\node[factor, below of = B] (D) {$f_3$};
	\node[factor, below of = C] (E) {$f_2$};
	
	\path (A) edge[thick] node {$y_1$} (B);
	\path (A) edge[thick] node[below=0pt] {$y_2$} (D);
	\path (B) edge[thick, double] node{$(x_1,x_1')$} (C);
	\path (E) edge[thick, double] node{$(x_2,x_2')$} (C);
	\path (E) edge[thick, double] node[below=0pt]{$(x_3,x_3')$} (D);
	\path (D) edge[thick, double] node{$(x_4,x_4')$} (B);
	\path (C) edge[thick] node{$y_3$} ([xshift=1.5cm]C);
	\path (E) edge[thick,double] node{$(x_5,x_5')$} ([xshift=1.5cm]E);
\end{tikzpicture}

%% file: figs2/partial_measurement_newnew.tikz
\begin{tikzpicture}[nodes = {draw= none, minimum size = 0pt},
	factor/.style={rectangle, minimum size=.5cm, draw, thick},
	largeFactor/.style={rectangle, minimum width = 0.6cm, minimum height = 1.2cm, draw, thick},
	node distance=1.5cm]
	\node[factor] (e) {}; \node at (e) {$\rho$}; 
	\node[right =1.5cm of e] (u0) {};
	\node[above=.25cm of u0, largeFactor] (u0a) {}; \node at (u0a) {$U_0$};
	\node[below=.25cm of u0, largeFactor] (u0b) {}; \node at (u0b)  {$U_0^\Herm$};
	\node[right=1.5cm of u0] (m) {};

	\node[above=.4cm of m, factor,minimum size=.6cm] (ma) {};
		\node[font=\small] at (ma) {$M_{y}$};
	\node[below=.4cm of m, factor,minimum size=.6cm] (mb) {};
		\node[font=\small] at (mb) {$M^\Herm_{y}$};
	
	\node[factor] at (m) (e1) {$=$};

	\node[right=1.5cm of m] (u1) {};
	\node[above=.25cm of u1, largeFactor] (u1a) {}; \node at (u1a) {$U_1$};
	\node[below=.25cm of u1, largeFactor] (u1b) {}; \node at (u1b)  {$U_1^\Herm$};
	\node[factor, right = 1.5cm of u1] (e2) {}; \node at (e2) {$=$};
	
	\path[draw,thick,*-] (e) |- node[above=0pt,font=\small,pos=.8]{$x_0$} (u0a); 
	\path[draw,thick,-*] (e) |- node[below=0pt,font=\small,pos=.8]{$x'_0$} (u0b); 
	\path[draw,thick,*-] (u0a.east|-ma.center) -- node[below=0pt,font=\small]{$x_3$} (ma.west);
	\path[draw,thick,-*] (u0b.east|-mb.center) -- node[above=0pt,font=\small]{$x_3'$} (mb.west);
	\path[draw,thick,*-] (ma.east) -- node[below=0pt,font=\small]{$x_4$} (u1a.west|-ma.center);
	\path[draw,thick,-*] (mb.east) -- node[above=0pt,font=\small]{$x_4'$} (u1b.west|-mb.center);
	\path[draw,thick,*-] ([yshift=.25cm]u0a.east) --
		node[above=0pt,font=\small]{$x_2$} ([yshift=.25cm]u1a.west);
	\path[draw,thick,-*] ([yshift=-.25cm]u0b.east) -- 
		node[below=0pt,font=\small]{$x_2'$}([yshift=-.25cm]u1b.west);
	\path[draw,thick,*-] (u1a) -| node[above,pos=.2,font=\small]{$x_5$} (e2);
	\path[draw,thick,-*] (u1b) -| node[below,pos=.2,font=\small]{$x_5'$} (e2);
	\path[draw,thick] (e1) -- (ma);
	\path[draw,thick] (e1) -- (mb);
	\path[draw,thick] (e1.east) -| node[right=0pt, pos=1] {$y$} ([xshift=.5cm,yshift=-1.8cm]e1.center);
\end{tikzpicture}

%% file: figs2/partial_measurement_defg_newnew.tikz
\begin{tikzpicture}[nodes = {draw= none, minimum size = 0pt},
	factor/.style={rectangle, minimum size=.5cm, draw, thick},
	largeFactor/.style={rectangle, minimum width = 0.6cm, minimum height = 1.2cm, draw, thick},
	node distance = 1.5cm]
	\node[factor] (e0) {};
	\node at (e0) {$\rho$};
	\node[largeFactor] (u0) [right =1.2cm of e0] {}; \node at (u0) {$\operator{U}_0$}; 
	\node[factor, draw=none, right = 1.2cm of u0] (fakem) {};
	\node[largeFactor] (u1) [right = 1.2cm of fakem] {}; \node at (u1) {$\operator{U}_1$}; 
	\node[factor, anchor = south, minimum size=.6cm] at (fakem.center|-u0.south) (m) {};
		\node[font=\small] at (m){$\operator{M}$};
	\node[factor,right = 1.2cm of u1] (e1) {};
		\node at (e1) {$I$};

	\path (e0) edge[draw,thick,double]
		node[above=0pt,font=\small]{$(x_0,x_0')$} (u0);
	\path (u0.east|-m.center) edge[draw,thick,double]
		node[above=0pt,font=\small]{$(x_3,x_3')$} (m);
	\path (m) edge[draw,thick,double]
		node[above=0pt,font=\small]{$(x_4,x_4')$} (m.center-|u1.west);
	\path ([yshift=.35cm]u0.east) edge[draw,thick,double] 
		node[above=0pt,font=\small]{$(x_2,x_2')$} ([yshift=.35cm]u1.west);
	\path (u1) edge[draw,thick,double] node[above=0pt,font=\small]{$(x_5,x_5')$} (e1);
	\path (m) edge[draw,thick] node[pos=1,right=0pt] {$y$} ([yshift=-.8cm]m);
\end{tikzpicture}